\documentclass[10pt, conference]{IEEEtran2}
\usepackage{epsf,psfrag,amssymb,amsfonts,amsmath,cite}
\usepackage{graphicx}
\def\HOME{D:/Ge}
\newtheorem{theorem}{Theorem}
\newtheorem{example}{Example}

\newtheorem{lemma}{Lemma}
\newtheorem{definition}{Definition}

\def\psfancypar#1#2{\begingroup\def\par{\endgraf\endgroup\lineskiplimit=0pt}
               \setbox2=\hbox{\large\sc #2}
               \newdimen\tmpht \tmpht \ht2 \advance\tmpht by \baselineskip
               \font\hhuge=Times-Bold at \tmpht
               \setbox1=\hbox{{\hhuge #1}}
               \count7=\tmpht \count8=\ht1
               \divide\count8 by 1000 \divide\count7 by \count8
               \tmpht=.001\tmpht\multiply\tmpht by \count7
               \font\hhuge=Times-Bold at \tmpht
               \setbox1=\hbox{{\hhuge #1}}
               \noindent
                \hangindent1.05\wd1
               \hangafter=-2 {\hskip-\hangindent
               \lower1\ht1\hbox{\raise1.0\ht2\copy1}%
                \kern-0\wd1}\copy2\lineskiplimit=-1000pt}

\newcommand{\beq}{\begin{equation}}
\newcommand{\eeq}{\end{equation}}
\newcommand{\bqa}{\begin{eqnarray}}
\newcommand{\eqa}{\end{eqnarray}}
\newcommand{\bqn}{\begin{eqnarray*}}
\newcommand{\eqn}{\end{eqnarray*}}
\newcommand{\nn}{\nonumber}

\newcommand{\be}{\begin{enumerate}}
\newcommand{\ee}{\end{enumerate}}
\newcommand{\bi}{\begin{itemize}}
\newcommand{\ei}{\end{itemize}}
\newcommand{\bd}{\begin{description}}
\newcommand{\ed}{\end{description}}
\newcommand{\ba}{\begin{array}}
\newcommand{\ea}{\end{array}}
\newcommand{\bde}{\begin{definition}}
\newcommand{\ede}{\end{definition}}
\newcommand{\bex}{\begin{example}}
\newcommand{\eex}{\end{example}}

\def\boxit#1{\vbox{\hrule\hbox{\vrule\kern3pt
        \vbox{\kern3pt#1\kern3pt}\kern3pt\vrule}\hrule}}

\def\reals{ { {\rm  I \kern-0.15em R }  } }
\def\complex{ {\,{{\rm C} \kern-0.50em \raise0.20ex {  |}}\, }}

\def\0bf{{\bf 0}}
\def\1bf{{\bf 1}}
\def\2bf{{\bf 2}}
\def\3bf{{\bf 3}}
\def\4bf{{\bf 4}}
\def\5bf{{\bf 5}}
\def\6bf{{\bf 6}}
\def\7bf{{\bf 7}}
\def\8bf{{\bf 8}}
\def\9bf{{\bf 9}}

\def\xbf{{\bf x}}
\def\ybf{{\bf y}}

\def\xbf{{\bf x}}
\def\ybf{{\bf y}}

\def\Rbf{{\bf R}}

\def\Wbf{{\bf W}}
\def\Xbf{{\bf X}}
\def\Ybf{{\bf Y}}

\def\Rmat{\mathcal{R}}

\def\Xmat{\mathcal{X}}
\def\Ymat{\mathcal{Y}}

\def\Rxx{\Rbf_{\ssstyle X\kern-.1em X}}

\let\ssstyle=\scriptscriptstyle

\def\Kout{\setbox1=\hbox{\Huge\bf K}\hbox to
1.05\wd1{\hspace{.05\wd1}
}}
\def\Sout{\setbox1=\hbox{\Huge\bf S}\hbox to 1.05\wd1{\hspace{.05\wd1}
}}

\def\scalefig#1{\epsfxsize #1\textwidth}
\begin{document}
\title{\huge The Sufficiency Principle for Decentralized Data Reduction}
%\IEEEoverridecommandlockouts
\author{\authorblockN{Ge Xu and Biao Chen}
\authorblockA{
Department of EECS\\
Syracuse University, NY, USA\\
 gexu\{bichen\}@syr.edu}}
\maketitle

\begin{abstract}
This paper develops the sufficiency principle  suitable for data reduction in decentralized inference systems. Both parallel and tandem networks are studied and we focus on the cases where observations at decentralized nodes are conditionally dependent. For a parallel network, through the introduction of a hidden variable that induces conditional independence among the observations, the locally sufficient statistics, defined with respect to the hidden variable, are shown to be globally sufficient for the parameter of inference interest. For a tandem network, the notion of conditional sufficiency  is introduced and the  related theories and tools are developed. Finally, connections between the sufficiency principle and some distributed source coding problems are explored.
\end{abstract}

\section{Introduction}
The sufficiency principle has played a prominent role in designing data processing methods for statistical inference. A sufficient statistic is a function of the data that contains all the information in the data about the parameter of interest. The primary goal of sufficiency-based data reduction is dimensionality reduction to facilitate subsequent inferences based on the reduced data
 \cite{Fisher:22, Casella&Berger:book,Lehmann&Casella:book}.

Suppose $\theta$ is the parameter of inference interest and $\Xbf\triangleq\{X_1,\cdots,X_n\}$ is a vector of random variables, whose distribution is given by $p(\xbf|\theta)
\footnote{We do not distinguish between probability density and probability mass function. Its meaning will become clear in the context of specific problems.}$.
If $T(\Xbf)$ is a sufficient statistic for $\theta$, then any inference about $\theta$ should depend on $\Xbf$ only through  $T(\Xbf)$\cite{Casella&Berger:book}. A useful tool to identify sufficient statistics is the Neyman-Fisher factorization theorem\cite[Theorem 6.2.6]{Casella&Berger:book} which states that  a statistic  $T(\Xbf)$ is sufficient for $\theta$ if and only if there exist functions $g(t|\theta)$ and $h(\xbf)$ such that %
 \bqn p(\xbf|\theta)=g(T(\xbf)|\theta)h(\xbf).\eqn
If the parameter $\theta$ is itself random, the sufficiency principle can also be reframed using the data processing inequality
\cite[Section 2.9]{Cover&Thomas:book2}. That is, a function $T(\Xbf)$ is a sufficient statistic if and only if the following Markov chain holds:
\bqn \theta-T(\Xbf)-\Xbf.\eqn

For decentralized inference, data reduction is done locally without access to the global data. Therefore, the contrasting notions of local sufficiency and global sufficiency  \cite{Viswanathan:93IT} need to be treated with care.
A sufficient statistic that is defined with respect to local data is referred to as locally sufficient statistic while a sufficient statistic defined with respect to the global data in the network is referred to as a globally sufficient statistic \cite{Viswanathan:93IT}.
As such, whether a statistic at a local node is globally sufficient  is not determined solely by the statistical characterization of local
data but also depends on the joint distribution of the whole data and how data/statistics are passed along within the network.

\begin{figure}
\centerline{
\begin{psfrags}
\psfrag{x}[c]{$\Xbf$}
\psfrag{y}[c]{$\Ybf$}
\psfrag{theta}[c]{$\theta$}
\psfrag{pxy}[c]{$p(\xbf,\ybf|\theta)$}
\psfrag{XX}[c]{$\mathbb{X}$ }
\psfrag{YY}[c]{$\mathbb{Y}$}
\psfrag{TX}[c]{$T(\Xbf)$}
\psfrag{TY}[c]{$T(\Ybf)$}
\psfrag{R}[c]{$\gamma(\cdot)$}
\psfrag{TH}[c]{$\hat{\theta}$}
 \scalefig{.5}\epsfbox{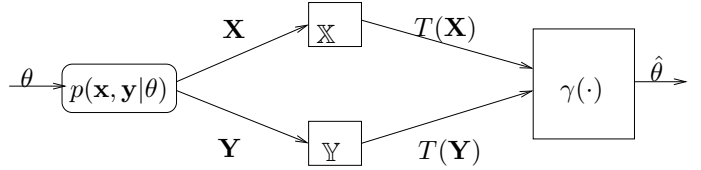}
\end{psfrags}}
\caption{\label{fig:1} Parallel network.}
\end{figure}
\begin{figure}
\centerline{
\begin{psfrags}
\psfrag{x}[c]{$\Xbf$}
\psfrag{y}[c]{$\Ybf$}
\psfrag{theta}[c]{$\theta$}
\psfrag{pxy}[c]{$p(\xbf,\ybf|\theta)$}
\psfrag{XX}[c]{$\mathbb{X}$ }
\psfrag{YY}[c]{$\mathbb{Y}$}
\psfrag{TX}[c]{$T(\Xbf)$}
\psfrag{R}[c]{$\gamma(\cdot)$}
\psfrag{TH}[c]{$\hat{\theta}$}
 \scalefig{.46}\epsfbox{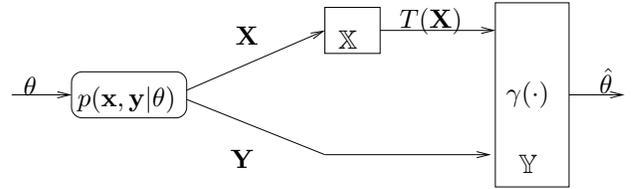}
\end{psfrags}}
\caption{\label{fig:2} Tandem network.}
\end{figure}
For conditionally independent observations (e.g., $\Xbf$ and $\Ybf$ are independent given $\theta$ in Figs. \ref{fig:1} and \ref{fig:2}), local sufficiency implies global sufficiency. This result was established in \cite{Hall&Wessel&Wise:91IT,Viswanathan:93IT,Ishwar:05JSAC} for parallel networks (Fig.~1) and it is straightforward to show that the same result holds for tandem networks (Fig.~2).
An interesting manifestation of the above result is in decentralized detection. It is well known that for a binary hypothesis testing problem, the likelihood ratio (LR) is a sufficient statistic for the underlying hypothesis.
Therefore, it is not surprising that likelihood ratio quantizers are globally optimal for decentralized detection with conditionally independent observations \cite{Tsitsiklis:93ASSP}, even with non-ideal, possibly coupling channels between the sensors and the fusion center \cite{Chen&Willett:05IT,Chen&Chen&Varshney:09IT}.

Without the conditional independence assumption, decentralized inference becomes considerably more complex.
For the decentralized detection, the optimal solution becomes NP complete when the observations
are  conditionally dependent\cite{Tsitsiklis&Athans:85AC}.
The primary focus of this paper is to develop theories and tools for decentralized data deduction with conditionally dependent observations for both parallel and tandem networks.

For parallel networks, we investigate the sufficiency principle under a hierarchical conditional independence (HCI) model, which is a new framework recently proposed to deal with distributed detection with conditionally dependent observations \cite{Chen&Chen&Varshney:12}.
The main idea is to inject a hidden variable $W$ such that the sensor observations are conditionally independent  with
respect to this new variable regardless of the dependence structure of the original model. Suitable conditions are identified under this HCI model such that local sufficiency implies global sufficiency.

For tandem networks such as that described in Fig. \ref{fig:2},  $\Ybf$ is fully available at the decision node.  As such, the novel notion of conditional sufficiency is defined to capture the difference in network structure with that of the parallel network. A new set of theories and tools corresponding to conditional sufficiency are then developed.

Finally, the developed notion of sufficiency is applied to some classical distributed source coding problems. There, sufficiency-based data reduction prior to a source encoder is shown to incur no penalty on the corresponding rate region or the rate distortion function.

The rest of the paper is organized as follows.  Section \ref{section 2} develops the sufficiency principle in parallel networks with conditionally dependent observations. Section \ref{section 3} deals with tandem networks where the notion of conditional sufficiency  is introduced and associated theories are developed.  In section \ref{section 4}, the connection between the developed sufficiency principle and two distributed source coding problems is explored. Section \ref{conclusion} concludes the paper.

\section{Sufficiency for Parallel Network}\label{section 2}
This section considers only a parallel network of two sensors as illustrated in Fig.~\ref{fig:1}. The result extends naturally to the case with arbitrary numbers of sensors.
Let  data available at node $\mathbb{X}$ be $\Xbf$ while data available at node $\mathbb{Y}$ be $\Ybf$.

Assume the parameter $\theta$ is random.
$(T_x(\Xbf),T_y(\Ybf))$ are globally sufficient for $\theta$ if the Markov chain $\theta-(T_x(\Xbf),T_y(\Ybf))-(\Xbf,\Ybf)$ holds.

Identifying local statistics that are globally sufficient can be accomplished in theory via the factorization theorem.
The process of using the factorization theorem may become cumbersome in a decentralized system or not applicable when the precise joint distribution of the data in the network is not available at local nodes. The following theorem provides certain relation between local sufficiency and global sufficiency
for a class of distributed inference problem.
%gives the relation of local sufficiency and global sufficiency under
%facilitate the process to find globally sufficient statistics.

\begin{lemma}\label{lemma 1}
Let $\Xbf,\Ybf\sim p(\xbf,\ybf|\theta)$ and suppose there exists a random variable $\Wbf$ such that
\bqa \theta-\Wbf-(\Xbf,\Ybf).\label{eq 1}\eqa
A statistic $T(\Xbf,\Ybf)$ that is sufficient  for $\Wbf$ is also sufficient for $\theta$.
\end{lemma}
\begin{proof}
The Markov chain (\ref{eq 1}) implies that  $\theta-\Wbf-(\Xbf,\Ybf,T(\Xbf,\Ybf))$  forms a Markov chain for any statistics $T(\Xbf,\Ybf)$.
That $T(\Xbf,\Ybf)$ is sufficient  for $\Wbf$ implies the Markov chain $\Wbf-T(\Xbf,\Ybf)-(\Xbf,\Ybf)$.
It is straightforward to show that these two Markov chains give rise to a  long Markov chain
\bqn\theta-\Wbf-T(\Xbf,\Ybf)-(\Xbf,\Ybf).\eqn
Therefore,  $T(\Xbf,\Ybf)$ is sufficient for $\theta$.
\end{proof}

Lemma \ref{lemma 1} is not useful in itself as $T(\Xbf,\Ybf)$ is a function of the global data which is not available in either of the nodes. Its use is main for establishing the following result.

\begin{theorem}\label{theorem 1}
Let $\Xbf,\Ybf\sim p(\xbf,\ybf|\theta)$ and suppose there exists a random variable $\Wbf$ such that
$\theta-\Wbf-(\Xbf,\Ybf)$.
Let $T(\Wbf)$ be a sufficient statistic for $\theta$, i.e., $\theta-T(\Wbf)-\Wbf$.
\begin{enumerate}
\item If a pair of statistics $(T_x(\Xbf),T_y(\Ybf))$  are globally sufficient  for $T(\Wbf)$,
they are globally sufficient  for $\theta$.

\item If $T(\Wbf)$ induces conditional independence between $\Xbf$ and $\Ybf$, and $(T_x(\Xbf),T_y(\Ybf))$ are locally sufficient for $T(\Wbf)$, then $(T_x(\Xbf),T_y(\Ybf))$ are globally sufficient for $\theta$.
\end{enumerate}
\end{theorem}
\begin{proof} To prove 1), from Lemma \ref{lemma 1}, we  only need to show that the Markov chain $\theta-T(\Wbf)-(\Xbf,\Ybf)$ holds. However, the Markov chain $T(\Wbf)-(\theta,\Wbf)-(\Xbf,\Ybf)$ together with $\theta-\Wbf-(\Xbf,\Ybf)$ results in the Markov chain $(\theta,T(\Wbf))-\Wbf-(\Xbf,\Ybf)$.
Combined with the Markov chain $\theta-T(\Wbf)-\Wbf$, we get $\theta-T(\Wbf)-\Wbf-(\Xbf,\Ybf)$ which implies $\theta-T(\Wbf)-(\Xbf,\Ybf)$.

For the second one, since conditional independence ensures that locally sufficient statistics are globally sufficient, $(T_x(\Xbf),T_y(\Ybf))$ are thus sufficient for $T(\Wbf)$. The first result then establishes that they are also sufficient for $\theta$.
\end{proof}

{\em Remark 1:}
It is given in \cite{Chen&Chen&Varshney:12} that any general distributed inference model can be represented as a HCI model and vice versa, where the HCI model is constructed by introducing a hidden variable $\Wbf$ such that the following Markov chains hold: $\theta-\Wbf-(\Xbf,\Ybf)$ and $\Xbf-\Wbf-\Ybf$.
Therefore, Theorem \ref{theorem 1} indicates that under the HCI model, local sufficiency with respect to the hidden variable  implies global sufficiency.

From the above result, it is clear that whether $T_x(\Xbf)$ is  globally sufficient depends also on $T_y(\Ybf)$ and vice versa. This coupling effect makes it rather difficult in studying the global sufficiency property. In the following, we consider a somewhat simplified situation where one is interested in data reduction at one node provided that a locally sufficient statistic from the other node is available at the fusion center. That is, if $T_y(\Ybf)$ is known to be a locally sufficient statistic, what should node $\mathbb{X}$
transmit such that $T_x(\Xbf)$ may form a globally sufficient statistic together with $T_y(\Ybf)$.
\begin{theorem}\label{theorem 2}
Let  $\Xbf,\Ybf$ be distributed according to $p(\xbf,\ybf|\theta)$. Assume $T_y(\Ybf)$ is a locally sufficient statistic for $\theta$,
then $(T_x(\Xbf),T_y(\Ybf))$ are globally sufficient for $\theta$ if and only if there exist functions $g(t_1|t_2,\theta)$ and $h(\xbf,\ybf)$ such that, for all sample points $(\xbf,\ybf)$ and all parameter values $\theta$, the conditional probability $p(\xbf|\ybf,\theta)$ satisfies
\bqa  p(\xbf|\ybf,\theta)=g(T_x(\xbf)|T_y(\ybf),\theta)h(\xbf,\ybf).\label{eq 2}\eqa
\end{theorem}
\vspace{0.1in}
\begin{proof}
Directly from the factorization theorem for $(\Xbf,\Ybf)$ and by rewriting $p(\xbf,\ybf|\theta)=p(\ybf|\theta)p(\xbf|\ybf,\theta)$.
\end{proof}

{\em Remark 2:}
Given a locally sufficient statistic $T_y(\Ybf)$, it is possible that there does not exist a $T_x(\Xbf)$ forming a globally sufficient statistic
together with $T_y(\Ybf)$.

{\em Remark 3:}
 The above result is shown under the assumption that $\theta$ is a random variable, similar result can be obtained for $\theta$ is not random by resorting to factorization theorem instead of data processing inequality.

\vspace{0.05in}
\begin{example}
 For $i=1,\cdots, n$,
let
\bqn X_i&=&Z+U_i\\ Y_i&=&Z+V_i,\eqn
 where $Z, U_1,\cdots, U_n,V_1,\cdots V_n$ are mutually independent Gaussian random variables such that $Z\sim N(\theta,\rho)$, $U_i\sim N(0,1-\rho), V_i\sim N(0,1-\rho)$. Thus, $X_i,Y_i\sim N(\theta,\theta,1,1,\rho)$.
The parameter of  inference interest is $\theta$.
 $\Xbf$ and $\Ybf$ are not conditionally independent given $\theta$.

Let $T(W)=W=Z$. Thus, $Z$ depends on $\theta$ through its mean value.
 Clearly, $Z$ satisfies the Markov chains $\theta-Z-(\Xbf,\Ybf)$
and $\Xbf-Z-\Ybf$ as required by the HCI model.
 Thus, from Theorem \ref{theorem 1}, the locally sufficient statistic pair for $Z$, $(\sum_iX_i,\sum_iY_i)$, is globally sufficient for $\theta$.
\end{example}
\vspace{0.05in}
\begin{example}
Consider the hypotheses test where the observations $X_i$, $i=1,\cdots,k$, satisfy the following model
\bqn
H_0:&& X_i=N_i,\\
H_1:&&X_i=h_iS+N_i,
\eqn
where $h_i$'s are complex Gaussian and independent of each other and of other variables, $S$ is a QAM signal taking values in the set
$s_m=r_me^{j\theta_m}$ with probability $\pi_m$ where $\theta_m=m\frac{2\pi}{M}$ for $m=1,\cdots,M$,  and $N_i$ is the independent observation noise at the $i$th sensor with $N_i\sim N(0,\sigma^2)$. The above model describes the problem of detecting the presence of a QAM signal in independent Rayleigh fading using $k$ sensors, e.g., as in cooperative spectrum sensing.
Each sensor
makes a local decision $U_i=\gamma(X_i)$ and sends it to a fusion center which makes a final decision regarding the hypothesis under test.

The observations are not conditionally independent given $H_1$.
 Let $W=S$ which induces conditional independence among observations under both hypotheses.
 It is easy to see that $T(W)=|S|$ is sufficient for $H$ given $S$.
Thus, the Markov chain $H-|S|-S-(X_1,\cdots,X_k)$ holds.

On the other hand, given $|S|$,  the observations are conditionally independent of each other under the QAM and Rayleigh fading assumptions.
For any $i$,
 $|X_i|$ is a minimal sufficient statistics for $|S|$. This can be easily verified by the ratio $\frac{p(x_i||s|)}{p(x_i'||s|)}$ for
two sample points $x_i $ and $x_i'$.
Therefore, by Theorem \ref{theorem 1}, $\{|X_i|\}$  is globally sufficient for $H$.

The above observation can be used to establish that the optimal detector at each local sensor is an energy detector for the model described in Example 2 \cite{Peng&Chen&Chen:12CISS}.
\end{example}

\section{Sufficiency for Tandem Network}\label{section 3}

A tandem network, as illustrated in Fig.~\ref{fig:2}, is one such that compressed data are transmitted to a node which also has its own observation. The second node will then make a final decision using its own data and the input from the first node. Knowing that $\Ybf$ is available at the fusion center even without directly observing $\Ybf$ should have an impact on how node $\mathbb{X}$ summarizes its own data $\Xbf$.
A natural way of extending the sufficiency principle  to this network is as follows: the inference performance should remain
the same whether the inference is based on $(\Xbf,\Ybf)$ or $(T(\Xbf),\Ybf)$.
From the data processing inequality, the sufficiency of $T(\Xbf)$ can thus be characterized using the Markov chain $\theta-(T(\Xbf),\Ybf)-(\Xbf,\Ybf)$. Given that $T(\Xbf)$ is a function $\Xbf$, it is straightforward to show that that the Markov chain  $\theta-(T(\Xbf),\Ybf)-(\Xbf,\Ybf)$ is equivalent to $\theta-(T(\Xbf),\Ybf)-\Xbf$.
This motivates the following definition of conditional sufficiency.
\begin{definition}\label{def CSS}
A statistic $T(\Xbf)$ is a {\em conditional sufficient statistic} for $\theta$, conditioned on $\Ybf$,
if the conditional distribution of the sample $\Xbf$ given the value of $T(\Xbf)$ and $\Ybf$ does not depend on $\theta$.
\end{definition}

The definition allows us to generalize a number of classical results related to sufficient statistics.

\begin{theorem}
Let  $\Xbf,\Ybf$ be distributed according to $p(\xbf,\ybf|\theta)$. Let $q(T(\xbf),\ybf|\theta)$ be the joint distribution of $T(\Xbf)$ and $\Ybf$, then
 $T(\Xbf)$ is a conditional sufficient statistic for $\theta$, conditioned on $\Ybf$,  if for every $(\xbf,\ybf)$ pair, the ratio $\frac{p(\xbf,\ybf|\theta)}{q(T(\xbf),\ybf|\theta)}$ is constant as a function of $\theta$.
\end{theorem}

Similarly, the Neyman-Fisher factorization theorem can also be generalized to the conditional case.
\begin{theorem}\label{theorem 4}
Let  $\Xbf,\Ybf$ be distributed according to $p(\xbf,\ybf|\theta)$. A statistic $T(\Xbf)$ is conditionally sufficient for $\theta$, conditioned on
$\Ybf$,  if and only if there exist functions $g(t,\ybf|\theta)$ and $h(\xbf,\ybf)$ such that,
\bqn p(\xbf,\ybf|\theta)=g(T(\xbf),\ybf|\theta)h(\xbf,\ybf),\eqn
 for all sample points $(\xbf,\ybf)$ and all parameter values $\theta$.
\end{theorem}
%\vspace{0.1in}

The proof can be constructed similarly to that of the factorization theorem in \cite[Theorem 6.2.6]{Casella&Berger:book}.
In fact, this result can be viewed as a special case of Theorem \ref{theorem 2} using the fact that $\Ybf$ is naturally a locally sufficient statistic for $\Ybf$.

{\em Remark 4:} For tandem networks, the definition of conditional sufficiency is more general than global sufficiency. This is because if
 there exist a  pair of statistics $(T_x(\Xbf),T_y(\Ybf))$ that are globally sufficient for $\theta$, then $T_x(\Xbf)$ must be conditionally sufficient for $\theta$, conditioned on $\Ybf$. %Similar statement can be said about $T_y(\Ybf)$. %However,  the other direction does not hold.
Therefore, for the inference problem under the HCI model, one can also obtain a conditional sufficient statistic using Theorem \ref{theorem 1}.

%\subsection{Minimal conditional sufficient Statistics}

%Minimal sufficient statistic plays a prominent role in statistical inference as it attains maximum data reduction without compromising inference performance.
Similar to the definition of minimal sufficient statistic\cite{Casella&Berger:book}, we can  define the notion of minimal conditional sufficient statistic as follows.
\begin{definition}
A conditional sufficient statistic $T(\Xbf)$ is a {\em minimal conditional sufficient statistic}  if it is a function of  any other conditional sufficient statistic $U(\Xbf)$.
\end{definition}

The following theorem provides a meaningful way to find minimal conditional sufficient statistics.
\begin{theorem}\label{Theorem 5}
Let $\Xbf,\Ybf$ be distributed according to $p(\xbf,\ybf|\theta)$. Suppose there exists a function $T(\xbf)$ such that for every two sample points
$\xbf$, $\hat{\xbf}$, and $\ybf$, the ratio $\frac{f(\xbf,\ybf|\theta)}{f(\hat{\xbf},\ybf|\theta)}$ is constant as a function of $\theta$ if and only if $T(\xbf)=T(\hat{\xbf})$. Then $T(\Xbf)$ is a minimal conditional sufficient statistic for $\theta$ given $\Ybf$.
\end{theorem}

The proof follows the same line of proof for Theorem 6.2.13 in \cite{Casella&Berger:book}.

\begin{example}

Let $\{X_i,Y_i\}$, $i=1,\cdots,n$  be independent and identically distributed (i.i.d) according to $p(x,y|\theta)$, where
\bqn p(x,y|\theta)=
\left\{\begin{array}{ll}
2 & \theta< x<\theta+1,\theta<y< x,\\
0 & \mbox{otherwise}.
\end{array}\right.\eqn
%The parameter of inference interest is $\theta$.
The marginal distribution of $X$ and $Y$ are therefore,
\bqn p(x|\theta)&=&2(x-\theta),~~\theta< x<\theta+1,\\
p(y|\theta)&=&2(\theta+1-y),~~\theta< y<\theta+1.\eqn

It can be easily shown that no data reduction is possible using the marginal distribution, i.e., no meaningful locally sufficient statistics can be found other than the data themselves.
Note that $X$ is uniformly distributed on the interval $(y,\theta+1)$, therefore, we have
\bqn p(\xbf|\ybf,\theta)=\frac{1}{\prod_{i=1}^n(\theta+1-y_i)}, y_i<x_i, (\max_{i}\{x_i\}-1)<\theta.\eqn
Thus, $\max_{i}\{X_i\}$ is a conditional sufficient statistic for $\theta$, conditioned on $\Ybf$. Similarly, we can obtain that
$\min_i \{Y_i\}$ is a conditional sufficient statistic of $\Ybf$, conditioned on the $X$ sequence. This
is consistent with the fact that
$(\max_{i}\{X_i\}, \min_{i}\{Y_i\})$
is globally sufficient given both $\Xbf$ and $\Ybf$.% sequences under the specified joint distribution.
\end{example}

\section{Sufficiency and Distributed Source Coding}\label{section 4}
For the point to point remote rate distortion problem, it was shown that sufficient statistic based data reduction achieves the same rate distortion function as the original data \cite{Eswaran&Gastpar:CISS05}.
 This section studies the connection between the sufficiency principle and distributed source coding problems. %the lossless source coding with side information problem and the remote source coding with side information available both at encoder and decoder.
\subsection{Source coding with side information}
 Consider the lossless source coding problem in Fig. \ref{fig:3}. An i.i.d. sequence of source pairs $(X^n,Y^n)$ are encoded separately with rates $(R_1,R_2)$ and the descriptions
  are sent to a decoder where only $X^n$ is to be recovered  with asymptotically vanishing probability of error. A rate pair $(R_1,R_2)$ is achievable if there exists a lossless source code with rates $(R_1,R_2)$. The rate region $\Rmat$ is defined as the closure of the set of all achievable rate pairs and was shown to be \cite{Ahlswede&Korner:75IT,Wyner:75IT},
  \bqn \Rmat=\{(R_1,R_2): R_1\geq H(X|U),
R_2\geq I(Y;U), X-Y-U\}.\eqn
%where $U$ is a auxiliary random variable with alphabet size $|\Umat|\leq |\Ymat|+1$.
 \begin{figure}
\centerline{
\begin{psfrags}
\psfrag{X}[c]{$X^n$}
\psfrag{Y}[c]{$Y^n$}
\psfrag{R1}[c]{$R_1$}
\psfrag{R2}[c]{$R_2$}
\psfrag{encoder1}[c]{Enc 1}
\psfrag{encoder2}[c]{Enc 2}
\psfrag{Dec}[c]{Dec}
\psfrag{hatX}[c]{$\hat{X}^n$}
 \scalefig{.4}\epsfbox{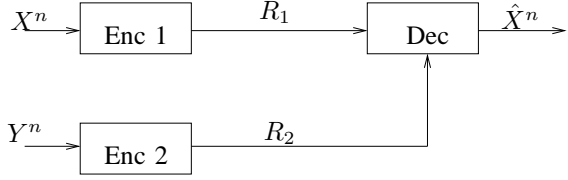}
\end{psfrags}}
\caption{\label{fig:3} Source coding with side information}
\end{figure}
  Assume $T(Y)$ is a sufficient statistic for $X$, i.e., $X-T(Y)-Y$.
Define \bqn \Rmat'=\{(R_1,R_2)&:& R_1\geq H(X|U), R_2\geq I(T(Y);U),\\&&X-T(Y)-U\},\eqn
which is the rate region for encoding $(X^n, T^n(Y^n))$ where $T^n(Y^n)$ is the i.i.d sequence $T(Y_i), i=1,\cdots,n$.
The following theorem shows that encoding reduced data $T^n(Y^n)$ achieves the same rate region as encoding the original data.
\begin{theorem}\label{theorem 6}
  \bqn \Rmat=\Rmat'\eqn
\end{theorem}
\vspace{0.1in}
\begin{proof}
 It is straightforward to show  $\Rmat\supseteq \Rmat'$. To show
 $\Rmat\subseteq \Rmat'$, let $(R_1,R_2)\in \Rmat$, then there exists a $U$ such that
$X-Y-U$, $R_1\geq H(X|U), R_2\geq I(Y;U)$. % there exists a random variable $V$ such that $X-T(Y)-V$, $R_1\geq H(X|V)$, $R_2\geq I(T(Y);V)$.
 Since $(X, T(Y))-Y-U$ and $X-T(Y)-Y$, the Markov chain $X-T(Y)-Y-U$ holds.
Therefore, $R_1\geq H(X|U)$, $R_2\geq I(Y; U)\geq I(T(Y); U)$ by the data processing inequality.
Thus, $(R_1,R_2)\in \Rmat'$.
\end{proof}

A direct consequence of Theorem \ref{theorem 6} is that the corner point of the rate region $(R_1=H(X|Y),R_2=H(T(Y))$ may
be strictly smaller than $(R_1=H(X|Y),R_2=H(Y)$. This observation was first reported in \cite{Kamath&Anantharam:10Allerton}.
Specifically, the corner point can be obtained by finding the smallest admissible $R_2$ when $R_1=H(X|Y)$ and it was shown
that \cite{Kamath&Anantharam:10Allerton}
\bqn \inf\{R_2:(H(X|Y), R_2)\in\Rmat\}\!\!\!&=&\!\!\!\inf_{X-Y-U,X-U-Y}I(Y;U),\\
&=&H(\Phi^X_Y).\eqn
As it turns out, the quantity $\Phi^X_Y$ is precisely the minimal sufficient statistic of $X$ given $Y$.
%of $I(Y;U)$ where the infimum is taken over all $U$ satisfying $X-$

\subsection{Remote source coding with side information}
%In this section, we examine the application of the conditional sufficient statistics in a remote rate distortion problem with side information.

Consider a model in Fig \ref{fig:4}, which is the remote  source coding with side information available at both the encoder and decoder. We will show that in this problem, the rate distortion function will not change by encoding a conditional sufficient statistic $T(X)$.% just like the one for the remote source coding problem considered in \cite[Appendix A]{Eswaran&Gastpar:CISS05}

Let $(X,Y,Z)\sim p(x,y,z)$ and $d(z,\hat{z})$ be a given distortion function. Let $(X^n,Y^n,Z^n)$ be i.i.d  sequences drawn
from $(X,Y,Z)$. Upon receiving the sequences $(X^n, Y^n)$, the encoder  generates
a description  of the sources with rate $R$  and sends it to the  decoder who has the side information  $Y^n$ and wishes to  reproduce $Z^n$ with distortion $D$.
%An $n$block code $(n,2^{nR},\Delta)$ can be defined in the usual way (e.g. in \cite{Yamamoto&Itoh:80IECE}) where
%the fidelity criterion
%is given by $ \Delta=E\frac{1}{n}\sum_{i=1}^nd(Z_i,\hat{Z}_i)$.
 The rate distortion function $R(D)$ is the infimum of rate $R$ such that there exist maps $f_n: \Xmat^n\times\Ymat^n\rightarrow\{1,\cdots, 2^{nR}\}$, $g_n: \Ymat^n\times\{1,\cdots, 2^{nR}\}\rightarrow\hat{Z}^n$ such that
\bqn  \limsup_{n\rightarrow \infty} Ed(Z^n,g_n(Y^n,f_n(X^n,Y^n)))\leq D.\eqn

It is easy to show that the rate distortion function $R(D)$ is:
\bqn R(D)=\min_{p(u|x,y)}\min_{f} I(X;U|Y),\eqn
where the minimum is taken over all $p(u|x,y)$ and functions $\hat{z}=f(u,y)$ such that
\bqa E_1[d(Z,\hat{Z})]=\sum_{x,y,z,u}p(x,y,z)p(u|x,y)d(z,f(u,y))\leq D.\label{eq 4}\eqa

Let $T(X)$ be a conditional sufficient statistic  for the remote source $Z$, conditioned on  $Y$ (i.e., $Z-(T(X),Y)-(X,Y)$).
%So $(T_i,Y_i,Z_i)$, $i=1,\cdots, n$, are i.i.d according to $p(t,y,z)$. Let $\Pmat_2$ be  the model in Figure \ref{fig:4} with $X^n$ replaced by $T^n(X^n)$. Similarly we have the rate distortion function for model $\Pmat_2$:
Define
\bqn R'(D)=\min_{p(u|t,y)}\min_fI(T(X);U|Y),\eqn
where the minimum is taken over all $p(u|t,y)$ and functions $\hat{z}=f(u,y)$ such that
%$E[D(Z,\hat{Z})]\leq d$.
\bqa E_2[d(Z,\hat{Z})]=\sum_{t,y,z,u}p(t,y,z)p(u|t,y)d(z,f(u,y))\leq D.\label{eq 5}\eqa
$R'(D)$ is the rate distortion function when we have $(T^n(X^n), Y^n)$ instead of $(X^n,Y^n)$ at the encoder, where $T^n(X^n)$ is the i.i.d sequence $T(X_i), i=1,\cdots,n$. %We have the following theorem.
\begin{theorem}
\bqn R(D)=R'(D).\eqn
\end{theorem}
\vspace{0.1in}
\begin{proof}
 It is obvious that $R(D)\leq R'(D)$.

We now show $R(D)\geq R'(D)$.
 For any $U$ that achieves $R(D)$, since $T(X)$ is a function of $X$,
 we have the Markov chain $(T(X),Y)-(X,Y)-U$, hence
 \bqn I(X;U|Y)\geq I(T(X);U|Y).\eqn

 Given that $T(X)$ is  a conditional sufficient statistic for $Z$, we have the following
 \bqa D &\geq&E_1[d(Z,\hat{Z})]\nn\\
%\!\!\!\!\! \!\!\!\!\!&=&\!\!\!\!\!\sum_{y,z,u}d(z,f(u,y))\left(\sum_xp(x,y,z)p(u|x,y)\right)\nn\\
 \!\!\!\!\! \!\!\!\!\!&=&\!\!\!\!\!\sum_{y,z,u}\!\!d(z,f(u,y))\!\!\left(\sum_xp(z|x,y)p(x,y,u)\!\!\right)\nn\\
   \!\!\!\!\! \!\!\!\!\!&=&\!\!\!\!\!\sum_{y,z,u}\!\!d(z,f(u,y))\!\!\left(\sum_tp(z|t,y)\sum_{x:T(x)=t}p(x,y,u)\!\!\right)\label{eq 6}\\
% where (\ref{eq 7}) is from the Markov chain $Z-(T(X),Y)-(X,Y)$. Define set $A_t\triangleq\{x: T(x)=t\}$, then
% \!\!\!\!\! \!\!\!\!\!&=&\!\!\!\!\!\sum_{y,z,u}\!\!d(z,f(u,y))\!\!\left(\sum_tp(z|t,y)\sum_{x:T(x)=t}p(x,y,u)\!\!\right) \\
  \!\!\!\!\! \!\!\!\!\!&=&\!\!\!\!\!\sum_{y,z,u}\!\!d(z,f(u,y))\!\!\left(\sum_tp(z|t,y)p(t,y,u)\!\!\right)\label{eq 7}
  %  \!\!\!\!\! \!\!\!\!\!&=&\!\!\!\!\!\E_2[d(Z,\hat{Z})]\nn
 \eqa
where (\ref{eq 6}) comes from the definition of conditional sufficiency and (\ref{eq 7}) is true by defining $p(t,y,u)=\sum_{x:T(x)=t}p(x,y,u)$.
This shows that for any $p(u|x,y)$ and $f(u,y)$  satisfying (\ref{eq 4})
there exist $p(u|t,y)$ and $f(u,y)$ such that (\ref{eq 5}) is satisfied.
Thus, $R(D)\geq R'(D)$.
\end{proof}

\begin{figure}
\centerline{
\begin{psfrags}
\psfrag{zn}[c]{$Z^n$}
\psfrag{X}[c]{$X^n$}
\psfrag{Y}[c]{$Y^n$}
\psfrag{p(xy|z)}[c]{$p(x,y|z)$}
\psfrag{R}[c]{$R$}
\psfrag{Enc}[c]{Enc}
\psfrag{Dec}[c]{Dec}
\psfrag{hz}[c]{$\hat{Z}^n$}
 \scalefig{.4}\epsfbox{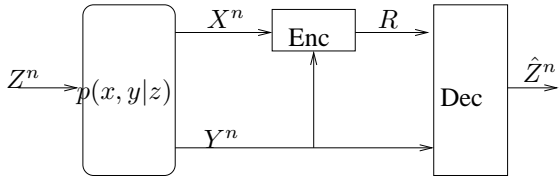}
\end{psfrags}}
\caption{\label{fig:4} Remote source coding with side information.}
\end{figure}
\section{Conclusion}\label{conclusion}

%Provided that the observations are conditionally independent, local sufficiency implies global sufficiency has been well known in decentralized inference.
 This paper developed the sufficiency principle that guides local data reduction in networked inference with dependent observations for two classes of inference networks: parallel network and tandem network.

For the parallel network, a previously proposed hierarchical conditional independence model is used to obtain conditions such that local sufficiency implies global sufficiency. A cooperative spectrum sensing example is given to illustrate the usefulness of such an approach.
For the tandem network, we introduced the notion of conditional sufficiency and developed related theories and tools.% associated with this new sufficiency concept.

The sufficiency principle for networked inference  has applications beyond that of decentralized inference. In particular, data reduction using suitable notions of sufficiency appears to incur no penalty on the rate region for various distributed source coding problem. There are potentially other distributed source coding problems where sufficiency based data reduction may also prove to be optimal.

\bibliographystyle{\HOME/tex/IEEEbib}
\bibliography{Ge}

\begin{thebibliography}{10}

\bibitem{Fisher:22}
R.~A. Fisher,
\newblock ``{On the mathematical foundations of theoretical statistics},''
\newblock in {\em { Philosophical Transactions of the Royal Society of London.
  Series A}}, 1922, vol. 222, pp. 309--368.

\bibitem{Casella&Berger:book}
G.~Casella and R.~L. Berger,
\newblock {\em Statistical Inference},
\newblock Duxbury, Belmont, CA, 1990.

\bibitem{Lehmann&Casella:book}
E.L. Lehmann and G.~Casella,
\newblock {\em Thoery of Point Estimation},
\newblock springer, New York, 2nd edition, 1998.

\bibitem{Cover&Thomas:book2}
T.M. Cover and J.A. Thomas,
\newblock {\em Elements of Information Theory},
\newblock Wiley, New York, 2nd edition, 2006.

\bibitem{Viswanathan:93IT}
R.~Viswanathan,
\newblock ``{A note on distributed estimation and sufficiency},''
\newblock {\em IEEE Trans.Inf. Theory}, vol. 39, no. 5, pp. 1765--1767, Sep.
  1993.

\bibitem{Hall&Wessel&Wise:91IT}
E.~B. Hall, A.~E. Wessel, and G.~L. Wise,
\newblock ``{Some aspects of fusion in estimation theory},''
\newblock {\em IEEE Trans. Inf. Theory}, vol. 37, pp. 420--422, 1991.

\bibitem{Ishwar:05JSAC}
P.~Ishwar, R.~Puri, K.~Ramchandran, and S.~S. Pradhan,
\newblock ``{On rate-constrained distributed estimation in unreliable sensor
  networks},''
\newblock {\em IEEE Journal on Seleted Areas in Communications}, pp. 765--775,
  April 2005.

\bibitem{Tsitsiklis:93ASSP}
J.~N. Tsitsiklis,
\newblock ``{Decentralized detection},''
\newblock in {\em Advances in Statistical Signal Processing}, H.~V. Poor and
  Eds. JAI~Press J.~B.~Thomas, Eds., Greenwich, CT, 1993.

\bibitem{Chen&Willett:05IT}
B.~Chen and P.K. Willett,
\newblock ``{On the optimality of likelihood ratio test for local sensor
  decisions in the presence of non-ideal channels},''
\newblock {\em IEEE Trans.Inf. Theory}, vol. 51, pp. 693--699, Feb. 2005.

\bibitem{Chen&Chen&Varshney:09IT}
H.~Chen, B.~Chen, and P.K. Varshney,
\newblock ``{Further results on the optimality of likelihood ratio quantizer
  for distributed detection in non-ideal channels},''
\newblock {\em IEEE Trans.Inf. Theory}, vol. 55, pp. 828--832, Feb. 2009.

\bibitem{Tsitsiklis&Athans:85AC}
J.N. Tsitsiklis and M.~Athans,
\newblock ``{On the complexity of decentralized decision making and detection
  problems},''
\newblock {\em IEEE Trans. on Automatic Control}, vol. 30, pp. 440--446, May
  1985.

\bibitem{Chen&Chen&Varshney:12}
H.~Chen, B.~Chen, and P.K. Varshney,
\newblock ``{A new framework for distributed detection with conditionally
  dependent observations},''
\newblock {\em IEEE Trans. Signal Processing}, vol. 60, no. 3, pp. 1409--1419,
  Mar. 2012.

\bibitem{Peng&Chen&Chen:12CISS}
F.~Peng, H.~Chen, and B.~Chen,
\newblock ``{On energy detection for cooperative spectrum sensing},''
\newblock in {\em Proc. of the 46th Conference on Information Sciences and
  Systems}, Princeton, NJ, March 2012.

\bibitem{Eswaran&Gastpar:CISS05}
K.~Eswaran and M.~Gastpar,
\newblock ``{Rate loss in the CEO problem},''
\newblock in {\em { Proc. of the 39th Conference on Information Sciences and
  Systems}}, Baltimore, MD, March 2005.

\bibitem{Ahlswede&Korner:75IT}
R.~F. Ahlswede and J.~K\"oner,
\newblock ``{Source coding with side information and a converse for degraded
  broadcast channels },''
\newblock {\em IEEE Trans. Inf. Theory}, vol. 21, pp. 629--637, May 1975.

\bibitem{Wyner:75IT}
A.~D. Wyner,
\newblock ``{On source coding with side information at the decoder},''
\newblock {\em IEEE Trans. Inf. Theory}, vol. 21, pp. 294--300, May 1975.

\bibitem{Kamath&Anantharam:10Allerton}
S.~Kamath and V.~Anantharam,
\newblock ``{A new dual to the G\'{a}cs-K\"{o}rner common information defined
  via the Gray-Wyner system},''
\newblock in {\em { Proc. Annual Allerton Conference on Communications, Control
  and Computing}}, Monticello, IL, Sep. 2010.

\end{thebibliography}
\end{document}